\documentclass[final,twocolumn]{IEEEtran}
\usepackage{amsthm,amssymb,mathptmx,amsbsy,amsmath,mathtools}
\usepackage{epstopdf}

\newtheorem{theorem}{Theorem}
\newtheorem{lemma}{Lemma}
\newtheorem{Definition}{Definition}
\newcommand{\mathcalbold}[1]{\ensuremath{\boldsymbol{\mathcal{#1}}}}

\author{Abla Kammoun}
\title{Robust blind methods using $\ell_p$ quasi-norms}

\begin{document}
\maketitle
\begin{abstract}
	It was shown in a previous work that some blind methods can be made robust to channel order overmodeling by using the $\ell_1$ or $\ell_p$ quasi-norms. However, no theoretical argument has been provided to support this statement. In this work, we study the robustness of subspace blind based methods using $\ell_1$ or $\ell_p$ quasi-norms. For the $\ell_1$ norm, we provide the sufficient and necessary condition that the channel should satisfy in order to ensure its identifiability in the noise-less case. We then study its frequency of occurrence, and deduce the effect of channel parameters on the robustness of blind subspace methods using $\ell_1$ norms. 
\end{abstract}
\section{Introduction}
Despite being proposed several decades ago, blind methods failed so far to be commercially deployed for civil communication products. This can be  attributed to several practical difficulties, thereby limiting their use to signal processing applications where training cannot be used. Among the most difficult challenges that need to be addressed is the sensitivity of a number of blind methods to the errors of channel over-modeling, \cite{karim-97,lisa-99}. This is for example the case of conventional subspace methods which are known to exhibit a high sensitivity in case of channel order over-modeling, even in the noiseless case. Actually, in the noiseless case, the channel is identified as the vector that spans the one dimensional kernel of a certain matrix ${\bf Q}$ which depends on the second order statistics of the received signal. But, when the channel order is over-estimated, the kernel of matrix ${\bf Q}$ is no longer a line but rather a vector space whose dimensions depend on the over-estimated order. The issue that should be raised is thus how to choose the right direction among all the vectors that span the kernel of matrix ${\bf Q}$.

To deal with this problem, a large effort has been investigated in either adding to sensitive methods a new feature that estimates efficiently the channel order \cite{gorokhov99} or proposing methods that are robust to channel order over-modeling. In this context, a new blind technique for sparse channel estimation has been recently proposed. To select the channel vector, this technique considers the joint utilization of a sparsity criterion which can be measured using a $\ell_p$ quasi norm ($0<p\leq 1$). It was noted by simulations that in this way, blind methods like cross relation and deterministic likelihood based techniques become robust to the over-estimation of the channel order \cite{bey08,bey08-b}. For non-sparse channels, introducing likewise a sparsity criterion shall enhance the channel identifiability probability. The reason is that over-modeling the channel is equivalent to zero-padding the channel vector, which artificially becomes sparse. Moreover, in the case of blind subspace methods, it can be shown that selecting the right direction is equivalent to choosing the vector that maximizes its sparsity. 

In this work, we propose to study the robustness of subspace based methods using $\ell_p$ quasi-norms for non-sparse channels. We derive necessary and sufficient condition for channel identifiability when  considering the $\ell_1$ norm as well as a sufficient condition if the $\ell_p$ quasi-norm is used. We then derive a lower bound on the probability that the necessary and sufficient condition   holds. This lower bound allows us to study the effect of the system parameters on the channel identifiability probability. For instance, we prove that increasing the number of antennas improves significantly the  channel identifiability probability, as opposed to increasing the number of channel coefficients, which tends to reduce it. 

We organize this paper as follows: we provide in section \ref{sec:model} the channel model. Section \ref{sec:identifiability} is devoted to the derivation of the channel identifiability conditions. Finally, a probabilistic analysis is performed before providing in section \ref{sec:simulation}, the simulation results. 
\section{System model and Blind subspace methods}
\label{sec:model}
Consider a SIMO communication link in which the receiver equipped with $M$ antennas receives data stemming from a single antenna transmitter. The received vector at time $k$ writes as:
$$
{\bf y}_k=\sum_{l=0}^L {\bf h}_ls_{k-l}+{\bf v}_k
$$
where ${\bf h}_l$ is the channel impulse response vector corresponding to the $l$-th tap and ${\bf v}_k$ denotes the additive Gaussian noise vector. 
Define ${\bf h}=\left[{\bf h}_0^{\mbox{\tiny T}},\cdots,{\bf h}_L^{\mbox{\tiny T}}\right]^{\mbox{\tiny T}}$ be the channel vector. 

Stacking $n$ observations of vector ${\bf y}$ in a $(n+1)M$ vector $\overline{\bf y}=\left[{\bf y}_k^{\mbox{\tiny T}},\cdots,{\bf y}_{k-n}^{\mbox{\tiny T}}\right]^{\mbox{\tiny T}}$, we will get:
$$
{\overline{\bf y}}=\mathcalbold{ I}_n({\bf h}){\bf s}_k+{\bf v}_k
$$
where $\mathcalbold{I}_n({\bf h})$ is the $M(n+1)\times (L+n+1)$ block-Toepltiz matrix:
$$
\mathcalbold{I}_n({\bf h})=\begin{bmatrix}
	{\bf h}_0 & \cdots & {\bf h}_L & & {\bf 0} \\
				       & {\bf h}_0 & \cdots & {\bf h}_L & \\
				       & & \ddots & \ddots & \\
	{\bf 0}& &{\bf h}_0& \cdots & {\bf h}_L
\end{bmatrix}
$$
Blind methods are based on the second order statistics of the received signal $\overline{\bf y}_n$. Assuming that:

{\bf A.1} The transmitter symbols $s_k$ are independent and identically distributed (i.i.d) with mean zero and variance $1$,

{\bf A.2} The undergone noise is white with variance $\sigma^2$,

The covariance matrix of $\overline{\bf y}_n$ is given by:
$$
{\bf R}=\mathbb{E}\overline{\bf y}_n\overline{\bf y}_n^{\mbox{\tiny T}}=\mathcalbold{I}_n({\bf h})\mathcalbold{I}_n({\bf h})^{\mbox{\tiny T}}+\sigma^2{\bf I}_{(n+1)M}.
$$
Under the assumption that 

{\bf A.3} The subchannels of vector ${\bf h}$ are real and have no zeros in common. Moreover, $n\geq L$ and $M> 1$, 

The rank of $\mathcalbold{I}_n({\bf h})$ is equal to $L+n+1$. Matrix ${\bf R}$ has then exactly $L+n+1$   eigenvalues which correspond to the signal subspace, whereas the remaining eigenvalues correspond to the noise subspace. It admits thus the following  singular value decomposition:
$$
{\bf R}={\bf U}\boldsymbol{\Lambda}{\bf U}^{\mbox{\tiny T}}+\sigma^2{\bf N}{\bf N}^{\mbox{\tiny T}}
$$
where $\Lambda$ is diagonal with $p=L+n+1$ nonzero diagonal elements and ${\bf U}$ and ${\bf N}$ are unitary matrices which span respectively the signal and noise spaces. 
Define $\boldsymbol{\Pi}={\bf N}{\bf N}^{\mbox{\tiny T}}$ the noise projector. In case of the absence of the noise term, the noise projector satisfies:
\begin{equation}
\boldsymbol{\Pi}\mathcalbold{I}_n({\bf h})={\bf 0}.
\label{eq:pi_h}
\end{equation}
Let $\boldsymbol{\Pi}=\left[\boldsymbol{\pi}_0,\cdots,\boldsymbol{\pi}_M\right]$. 
Subspace based methods consists in searching the solution of the following equation:
\begin{equation}
\boldsymbol{\Pi}\mathcalbold{I}_N({\bf f})={\bf 0}
\label{eq:f}
\end{equation}

If the number of coefficients $L$ in ${\bf f}$ is accurately known, then there exists up to a constant term only one vector ${\bf f}$ of length $L$ satisfying \eqref{eq:f}. If $L'\geq L$, then, the solutions of $\eqref{eq:f}$ belong to a vector space of dimension $L'-L+1$ spanned by the columns of the following $M(L'+1)\times L'-L+1$ Toeplitz matrix \cite[Lemma2]{karim-97}:

$$
{\bf H}=\begin{bmatrix}
{\bf h}_0 & \cdots & {\bf 0}\\
\vdots & \ddots & \vdots\\
{\bf h}_L & \ddots & {\bf h}_0 \\
\vdots & \ddots & \vdots \\
{\bf 0} & \cdots & {\bf h}_L\\
 \end{bmatrix}
$$
Obviously, the channel vector ${\bf h}$ is almost surely the vector that exhibits the highest sparsity among all vectors in the range space of matrix ${\bf H}$. In case the channel order is unknown, the channel vector can be thus recovered by solving the following optimization problem:
\begin{equation}
	\begin{aligned}
		& \underset{{\bf f}}{\text{minimize}}
		& & \|{\bf f}\|_0 \\
	& \text{subject to}
		& & {\bf f}={\bf H}{\bf s}\\
  		& & & s_1 =1
	\end{aligned}
	\label{eq:ProblemP0}
	\tag{$P_0$}
\end{equation}
where the condition ${s}_1=1$ is introduced to exclude the trivial zero solution ($s_1$ refers to the first entry of vector ${\bf s}$.).
Solving \eqref{eq:ProblemP0} requires an intractable combinatorial search, thereby reducing its interest for real-time applications. A common used way is to substitute the quasi norm $\ell_0$ which is non-convex and non-continuous by the $\ell_p$ quasi norm, thereby leading to the following optimization problem:
\begin{equation}
	\begin{aligned}
		& \underset{{\bf f}}{\text{minimize}}
		& & \|{\bf f}\|_p \\
	& \text{subject to}
		& & {\bf f}={\bf H}{\bf s}\\
		& & & s_1=1
	\end{aligned}
	\label{eq:ProblemPp}
	\tag{$P_p$}
	\end{equation}
	where $\|{\bf x}\|_p^p=\sum_i |x_i|^p$.
\section{Channel Identifiability Conditions}
\label{sec:identifiability}
It has been shown by using simulations in \cite{bey08-b}, that the use of $\ell_1$ and $\ell_p$ norms can enhance the robustness of blind subspace methods. The aim of this paper is to provide theoretical arguments that account for the observed aspect. More explicitly, we will provide necessary and/or sufficient conditions that imply that the solution of problem $(P_p)$ correspond to the zero-padded channel vector in the noiseless case, a case which we refer to as channel identifiability.

Before that, we shall introduce some matrices depending on the channel, which will serve later to formulate the identifiability conditions. Since $s_1$ has to be set to $1$, problem $(P_p)$ could be written as:
\begin{equation}
	\begin{aligned}
		& \underset{{\bf g}\in\mathbb{R}^{L'-L}}{\text{minimize}}
		& & \|{\bf h}_z+\widetilde{\bf H}g\|_p \\
	\end{aligned}
	\tag{$P_p$}
	\end{equation}
where ${\bf h}_z=\left[{\bf h}_1^{\mbox{\tiny T}},\cdots,{\bf h}_L^{\mbox{\tiny T}}\right]^{\mbox{\tiny T}}$  and $\widetilde{\bf H}$ is the $ML'\times L'-L$ block-Toeplitz matrix having the same shape as ${\bf H}$. Partition $\widetilde{\bf H}$ as $$\widetilde{\bf H}=\left[\begin{array}{c}{\bf A}\\\hline{\bf B} \end{array}\right],$$ where ${\bf A}$ is formed by selecting the first $ML$ rows of $\widetilde{\bf H}$ whereas ${\bf B}$ correspond to the sub-matrix composed of the last $M(L'-L)$ rows of $\widetilde{\bf H}$, i.e. ${\bf A}$ and ${\bf B}$ are given by:
	\begin{align*}
{\bf B}&=
\begin{bmatrix}
 {\bf h}_L & {\bf h}_{L-1} & \cdots & {\bf h}_{L-\delta+1}\\
{\bf 0} & {\bf h}_L & & \vdots \\
\vdots & \ddots & \ddots &\vdots\\
{\bf 0} & \cdots & {\bf 0} & {\bf h}_L
\end{bmatrix}\\
{\bf A}&=
\begin{bmatrix}
 {\bf h}_0 & {\bf 0}& \cdots& {\bf 0} \\
{\bf h}_1 & \ddots & & \\
\vdots & \ddots & &{\bf h}_0\\
\vdots & \vdots & &\vdots \\
{\bf h}_{L-1} & {\bf h}_{L-2}& \cdots & {\bf h}_{L-\delta}
\end{bmatrix}.
\end{align*}
\subsection{Necessary and sufficient condition for $\ell_1$ norm}
Unlike the $\ell_p$ quasi-norm ($p<1$), the $\ell_1$ norm is convex. It is thus possible to derive necessary and sufficient condition that implies channel identifiability.

Our theorem is stated as follows:
\begin{theorem}{Necessary and sufficient condition}

Let ${\bf v}=\left[{\rm sign}({\bf h}_1)^{\mbox{\tiny T}},\cdots,{\rm sign}({\bf h}_L)^{\mbox{\tiny T}}\right]^{\mbox{\tiny T}}$ and assume that $L > L'-L \geq 1$. Then the necessary and sufficient condition for channel identifiability can be expressed as :
\begin{equation}
\frac{\left|{\bf v}^{\mbox{\tiny T}}{\bf A}{\bf g}\right|}{\|{\bf Bg}\|_1}\leq 1  \ \  \forall {\bf g}\ \ \in\mathbb{R}^{L'-L}
\label{eq:necessary}
\end{equation}
\label{th:necessary}
\end{theorem}

\begin{proof}
	See Appendix \ref{app:necessary}
\end{proof}

\subsection{Sufficient condition for $\ell_p$ quasi-norm}
Since the $\ell_p$ quasi-norm is a non-convex function, the problem might have many local minima. Nevertheless, we still can find a sufficient condition that ensures that the channel can be identified as a local minimum of problem $(P_p)$. The result is stated  in the following theorem:
\begin{theorem}
Let ${\bf v}=p\begin{bmatrix}{\rm sign}({\bf h}_1) \\ \vdots \\ {\rm sign}({\bf h}_L)\end{bmatrix}\bullet  \begin{bmatrix}|{\bf h}_1|^{p-1}\\ \vdots |{\bf h}_L|^{p-1}
\end{bmatrix}$  where $\bullet$ denotes the Hadamard (element by element product). If the following condition is satisfied:
\begin{equation}
\frac{\left|{\bf v}^{\mbox{\tiny T}}{\bf A}{\bf g}\right|}{\|{\bf Bg}\|_1}\leq 1  \ \  \forall {\bf g}\ \ \in\mathbb{R}^{L'-L}
\label{eq:sufficient}
\end{equation}
Then, the channel can be identified as a local minimum of $(P_p)$.
\label{th:sufficient}
\end{theorem}
\begin{proof}
See Appendix \ref{app:sufficient}
\end{proof}
{\bf Remark} Note that for each channel realization, there exists $p$ such that the channel identifiability condition \eqref{eq:sufficient} occurs.  
\section{Probabilistic Analysis}
The necessary and sufficient conditions being provided, a natural question that arises is how often they occur. Obviously, this will depend on the distribution of the channel vector elements which have to be chosen. We will consider next the assumption that the channel vector elements are drawn from the Gaussian distribution with zero mean and variance $\frac{1}{L+1}$ :

{\bf A.4} The entries of the channel vectors ${\bf h}_1,\cdots, {\bf h}_L$ are real Gaussian with zero mean and variance $\frac{1}{L+1}$. 

To determine a lower bound on the channel identifiability probability, we will rely on the techniques derived in \cite{gribonval10,gribonval08}. Actually, in the same way as in \cite{gribonval10}, we notice that the necessary and sufficient conditions in the $\ell_1$ and $\ell_p$ norm possesses an equivalent characterization  as shown in the following lemma whose proof can be found in \cite{gribonval10}:
\begin{lemma}
	Let ${\bf B}$ be an $n\times N$ matrix with rank $N$. For any vector $z$ define:
	$$
	\|z\|_B:= \sup_{{\bf x}\neq 0} \frac{\left|{\bf z}^{\mbox{\tiny T}}{\bf x}\right|}{\|{\bf B}{\bf x}\|_1}
	$$
	We have then the equivalent characterization :
	$$
	\|{\bf z}\|_B= \min \|{\bf d}\|_{\infty} \ \ \textnormal{under the constraint} \ \ {\bf B}^{\mbox{\tiny T}}{\bf d} = {\bf z}
	$$
	\label{lemma:b}
\end{lemma}
Applying lemma \ref{lemma:b}, conditions \eqref{eq:sufficient} and \eqref{eq:necessary} are equivalent to stating that the solution of the following problem :
$$
\min \|{\bf d}\|_\infty \ \  \textnormal{under the constraint} \ \  {\bf B}^{\mbox{\tiny T}}{\bf d} = {\bf A}^{\mbox{\tiny T}} {\bf v}
$$
achieves a minimum value which is less than $1$. In other words, the necessary and sufficient conditions are satisfied if and only if there exists a vector ${\bf d}$ with $\|{\bf d}\|_{\infty} \leq 1$ such that ${\bf B}^{\mbox{\tiny T}}{\bf d}={\bf A}^{\mbox{\tiny T}}{\bf v}$. 

Since ${\rm rank}({\bf B})=\delta$ almost surely, the channel identifiability will hold if the following conditions are satisfied:
\begin{itemize}
	\item The image of the cube by matrix ${\bf B}^{\mbox{\tiny T}}$ contains a ball of radius $\alpha$,
	\item The vector ${\bf A}^{\mbox{\tiny T}}{\bf v}$ satisfies $\|{\bf A}^{\mbox{\tiny T}}{\bf v}\|_2\leq \alpha$. 
	\end{itemize}
	Let $\mathcal{P}$ denote the probability that the channel identifiability holds, and $E_{\alpha}^1$ and $E_{\alpha}^2$ be the events given by:
	\begin{align*}
		E_{\alpha}^1&=\left\{\textnormal{The image of the cube by} \ \ {\bf B}^{\mbox{\tiny T}} \ \ \textnormal{contains a ball of radius} \ \ \alpha\right\} \\
		E_{\alpha}^2&=\left\{\|{\bf A}^{\mbox{\tiny T}}{\bf v}\|\leq \alpha\right\}
	\end{align*}
	Then, $\mathcal{P}$ can be lower-bounded as :
	$$
	\mathcal{P}\geq \mathbb{P}\left\{\bigcup_{\alpha}E_{\alpha}^1\cap E_{\alpha}^2\right\}\geq \max_{\alpha}\mathbb{P}\left(E_{\alpha}^1\cap E_{\alpha}^2\right).$$
Next, we will consider the computation of the lower bound for the $\ell_1$ norm when $\delta=1$. More explicitly, we show the following theorem:
\begin{theorem}
	For $\delta=1$, the probability $\mathcal{P}$ that the necessary and sufficient condition \eqref{eq:necessary} occurs is greater than:
\begin{equation}
	\mathcal{P} \geq \max_{\epsilon \in \left[0,1\right]}\left(1-\exp(-\frac{M\epsilon^2}{\pi})\right)\frac{1}{\sqrt{\pi}}\gamma(\frac{1}{2},\frac{M(1-\epsilon)^2}{\pi L})
	\label{eq:probability}
\end{equation}
\end{theorem}
\begin{proof}
	When $\delta=1$, it is easy to note that the events $E_{\alpha}^1$ and $E_{\alpha}^2$ are independent. Therefore, 
	$$
	\mathbb{P}(E_{\alpha}^1\cap E_{\alpha}^2)=\mathbb{P}(E_{\alpha}^1)\mathbb{P}(E_{\alpha}^2)
	$$
	On the other hand, ${\bf v}^{\mbox{\tiny T}}{\bf A}$ is a scalar Gaussian random variable with mean $0$ and variance $\frac{LM}{L+1}$. Therefore,
	\begin{align*}
		\mathbb{P}\left[\left|{\bf v}^{\mbox{\tiny T}}{\bf A}\right|\leq \alpha\right]&=\mathbb{P}\left[\left|{\bf v}^{\mbox{\tiny T}}{\bf A}\right|^2\frac{(L+1)}{LM}\leq \alpha^2\frac{(L+1)}{LM}\right] \\				     &=\frac{1}{\sqrt{\pi}}\gamma\left(\frac{1}{2},\frac{\alpha^2(L+1)}{2LM}\right)
	\end{align*}
	On the other hand, using standard concentration inequalities (See Theorem 2 in \cite{Gribonval-07}), we have:
	$$
	\mathbb{P}\left[\sqrt{L+1}\|{\bf h}_L\|_1\leq (1-\epsilon)\sqrt{\frac{2}{\pi}}M\right] \leq \exp\left(-\frac{\epsilon^2 M}{\pi}\right)
	$$
	Therefore, 
	$$
	\mathbb{P}\left[\sqrt{L+1}\|{\bf h}_L\|_1\leq (1-\epsilon)\sqrt{\frac{2}{\pi}}M\right] \leq \exp\left(-\frac{\epsilon^2 M}{\pi}\right)
	$$
As a consequence,
	$$
	\mathbb{P}\left[\|{\bf h}_L\|_1\leq (1-\epsilon)\sqrt{\frac{2}{\pi(L+1)}}M\right] \leq \exp\left(-\frac{\epsilon^2 M}{\pi}\right)
	$$
	This gives thus the following probability lower bound :
	\begin{align*}
		\mathbb{P}\left[\|{\bf h}_L\|_1\geq (1-\epsilon)\sqrt{\frac{2}{\pi(L+1)}} M\right] \geq 1-\exp\left(-\frac{M\epsilon^2}{\pi}\right)
\end{align*}
Taking $\alpha=(1-\epsilon)\sqrt{\frac{2}{\pi(L+1)}}M$, we get :
$$
\mathcal{P}\geq \max_{\epsilon\in\left[0,1\right]} \left(1-\exp\left(-\frac{M\epsilon^2}{\pi}\right)\right)\frac{1}{\sqrt{\pi}}\gamma(\frac{1}{2},\frac{M(1-\epsilon)^2}{\pi L})
$$
\end{proof}
{\bf Interpretation :}
From the result shown above, one can note that as the number of antennas increases, the probability that the channel identifiability condition occurs tends to $1$. On the other side, the order of the channel $L$ might tend to have the opposite effect : it tends to make it fail when it increases. 
\section{Simulation Results}
\label{sec:simulation}
We present here simulation results for the $\ell_1$ norm. Fig. \ref{fig:prob} displays the effect of the system parameters $L$ and $M$ on the lower bound probability that we have computed by maximizing numerically \eqref{eq:probability}. Fig. \ref{fig:prob} shows that increasing the number of antennas tends to enhance the channel identifiability, a result which has been previously mentioned.
\begin{figure}[h]
	\begin{center}
	\includegraphics[scale=0.5]{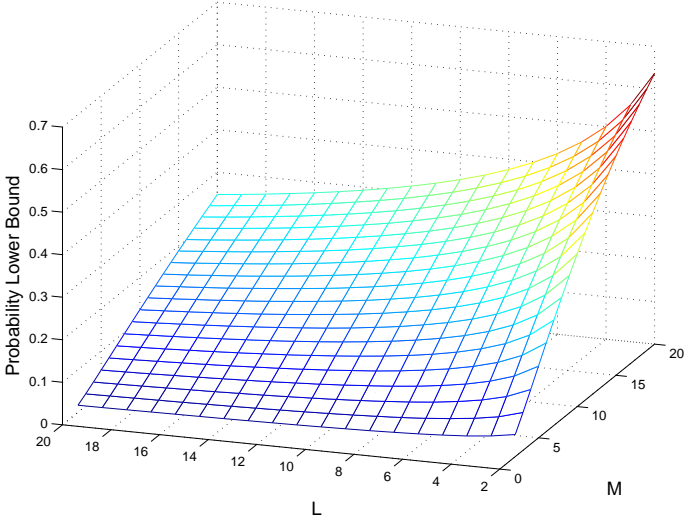}
\end{center}
	\caption{Impact of the system parameters $L$ and $M$ on the lower bound probability}
	\label{fig:prob}
\end{figure}



\appendices
\section{Proof of Theorem \ref{th:necessary}}
\label{app:necessary}
The proof of theorem \ref{th:necessary} will rely on the following mathematical results about the optimization of convex functions, which we provide for sake of completeness. 
\begin{Definition}
	Let $f: \mathbb{R}^n\mapsto \mathbb{R}$ be a real-valued function. The directional derivative of $f$ at ${\bf x}_0$ is given by : 
	$$
	f'(x_0,y)= \inf_{t>0} \frac{f({\bf x}_0+ty)-f({\bf x}_0)}{t}
	$$
	\label{def:directional_derivative}
\end{Definition}
\begin{theorem}
	Let $f$ be a convex function defined on a convex set $X$ and ${\bf x}_0\in X$ be a point where $f$ is finite. Then ${\bf x}_0$ is a global minimum point of $f$ if and only of the following condition holds:
	$$
	f'({\bf x}_0,{\bf x}-{\bf x}_0)\geq 0 \ \  \forall \ \ {\bf x}\in X.
	$$
\end{theorem}
We are now in position to establish theorem \ref{th:necessary}. 

Let $f$ be the convex function defined as :
$$
f({\bf g})=\|{\bf h}_z+\widetilde{\bf H}{\bf g}\|_1.
$$
Then, $f({\bf g})$ can be also given by: 
$$
f({\bf g})=\|\widetilde{\bf h}+{\bf Ag}\|_1+\|{\bf Bg}\|_1
$$
where $\widetilde{\bf h}=\left[{\bf h}_1^{\mbox{\tiny T}},\cdots,{\bf h}_L^{\mbox{\tiny T}}\right]^{\mbox{\tiny T}}$.


It is clear that channel identifiability condition occurs if ${\bf g}={\bf 0}$ is a global minimum. By definition \ref{def:directional_derivative}, this implies that
\begin{equation}
f'(0,{\bf y})\geq 0 \ \ \forall \ \ {\bf y}\in \mathbb{R}^{L'-L}
\label{eq:condition}
\end{equation}
Since
\begin{align*}
	f'(0,{\bf g})&=\inf_{t>0} \frac{\|\widetilde{\bf h}+t{\bf Ag}\|_1+t\|{\bf Bg}\|_1-{\|\widetilde{\bf h}\|_1}}{t} \\
		    &={\bf v}^{\mbox{\tiny T}}{\bf Ag}+\|{\bf By}\|_1. 
\end{align*}
equation \eqref{eq:condition} implies that:
$$
\frac{\left|{\bf v}^{\mbox{\tiny T}}{\bf A}{\bf g}\right|}{\|{\bf Bg}\|_1}\leq 1.
$$
\section{Proof Theorem \ref{th:sufficient}}
\label{app:sufficient}
To prove theorem \ref{th:sufficient}, it suffices to establish that when condition \eqref{eq:sufficient} is satisfied, vector ${\bf g}={\bf 0}$ is a local minimum of the function
$$
f_p({\bf g})=\|{\bf h}_z+\widetilde{\bf H}{\bf g}\|_p^p=\|\widetilde{\bf h}+{\bf Ag}\|_p^p+\|{\bf Bg}\|_p^p
$$
where $\widetilde{\bf h}=\left[{\bf h}_1^{\mbox{\tiny T}},\cdots,{\bf h}_L^{\mbox{\tiny T}}\right]^{\mbox{\tiny T}}$.
In other words, we need to show that for all $\epsilon>0$ there exists a neighborhood in which all vectors ${\bf g}$ satisfy $\|\widetilde{\bf h}+{\bf A g}\|_p^p+\|{\bf Bg}\|_p^p \geq \|\bf h\|_p^p$. 

Let $\widetilde{h}_i$ denote the $i-$th entry of ${\bf h}$ and ${\bf e}_i$ be the vector with $1$ on the $i-$th position and $0$ elsewhere. Let $\rho_i={\rm sign}(\widetilde{h}_i)$. Then, it can be easy to show that :
\begin{align*}
	\left|\widetilde{h}_i+{\bf e}_i^{\mbox{\tiny T}}{\bf Ag}\right|^p&=\left||\widetilde{h}_i|+\rho_i{\bf e}_i^{\mbox{\tiny T}}{\bf Ag}\right|^p \\				 &=|\widetilde{h}_i|^p\left|1+\frac{\rho_i{\bf e}_i^{\mbox{\tiny T}}{\bf Ag}}{|\widetilde{h}_i|}\right|^p
\end{align*}
Thereby, the Taylor expansion of $\left|\widetilde{h}_i+{\bf e}_i^{\mbox{\tiny T}}{\bf Ag}\right|^p$ when ${\bf g}$ lies in a neighborhood of ${\bf 0}$ can be given by :
$$
\left|\widetilde{h}_i+{\bf e}_i^{\mbox{\tiny T}}{\bf Ag}\right|^p=|\widetilde{h}_i|^p\left(1+\frac{p\rho_i{\bf e}_i^{\mbox{\tiny T}}{\bf Ag}}{|\widetilde{h}_i|}+\frac{p(p-1)}{2}\frac{|{\bf e}_i^{\mbox{\tiny T}}{\bf A g}|^2}{|\widetilde{h}_i|^2}\right)+o(\|{\bf g}\|^3)
$$
Hence,
$$
\frac{|\widetilde{h}_i+{\bf e}_i^{\mbox{\tiny T}}{\bf Ag}|^p-|\widetilde{h}_i|^p-p\rho_i|\widetilde{h}_i|^{p-1}{\bf e}_i^{\mbox{\tiny T}}{\bf Ag }}{p\rho_i|\widetilde{h}_i|^{p-1}{\bf e}_i^{\mbox{\tiny T}}{\bf Ag}}\xrightarrow[{\bf g}\to {\bf 0}]{} 0.
$$
As a consequence, for all $\epsilon>0$, there exists a neighborhood $V_{\epsilon,i}$ of $0$ such that : 
$$
\left|\widetilde{h}_i+{\bf e}_i^{\mbox{\tiny T}}{\bf Ag}\right|^p \geq |\widetilde{h}_i|^p+p|\widetilde{h}_i|^{p-1}\rho_i{\bf e}_i^{\mbox{\tiny T}}{\bf Ag}-\epsilon p |\widetilde{h}_i|^{p-1}|{\bf e}_i^{\mbox{\tiny T}}{\bf Ag}|
$$
We can assume without loss of generality that:
$$
\sum_{i=1}^{LM}\rho_i|\widetilde{h}_i|^{p-1}{\bf e}_i^{\mbox{\tiny T}}{\bf Ag}\leq 0.
$$
Because, otherwise, one can prove easily for $\epsilon$ small enough, the existence of a neighborhood  in which the following desired inequality holds
$$
\|\widetilde{\bf h}+{\bf Ag}\|_p^p+\|{\bf Bg}\|_p^p \geq \|\widetilde{\bf h}\|_p^p.
$$
Since $\sum_{i=1}^{LM}p|\widetilde{h}_i|^{p-1}\rho_i{\bf e}_i^{\mbox{\tiny T}}{\bf Ag}=-\left|\sum_{i=1}^{LM}p |\widetilde{h}_i|^{p-1}\rho_i{\bf e}_i^{\mbox{\tiny T}}{\bf Ag}\right|,$
we get for all ${\bf g}\in \widetilde{V}_{\epsilon}=\displaystyle{\cap_{i=1,\cdots,LM}V_{\epsilon,i}}$, 
\begin{equation}
\sum_{i=1}^{LM}\left||\widetilde{h}_i|+\rho_i{\bf e}_i^{\mbox{\tiny T}}{\bf Ag}\right|^p \geq \sum_{i=1}^{LM} |\widetilde{h}_i|^p -\left|\sum_{i=1}^{LM}p|\widetilde{h}_i|^{p-1}\rho_i{\bf e}_i^{\mbox{\tiny T}}{\bf Ag}\right|-\epsilon p \sum_{i=1}^{LM} |\widetilde{h}_i|^{p-1}|{\bf e}_i^{\mbox{\tiny T}}{\bf Ag}|
\label{eq:ineq_f}
\end{equation}
On the other hand, as $p<1$, one can easily check that
$$
\frac{\|{\bf Bg}\|_p^p}{\|{\bf Bg}\|_1-\epsilon p \sum_{i=1}^{LM}|\widetilde{h}_i|^{p-1}|{\bf e}_i^{\mbox{\tiny T}}{\bf Ag}|} \xrightarrow[{\bf g}\to 0]{} +\infty
$$
Choose $\epsilon$ such that 
$$
\|{\bf Bg}\|_1\geq \epsilon p \sum_{i=1}^{LM}|\widetilde{h}_i|^{p-1}|{\bf e}_i^{\mbox{\tiny T}}{\bf Ag }|.
$$
This is possible since, with probability equal to one,  $\|{\bf Bg}\|=0$ if and only if ${\bf g}=0$. 
Hence, there exists a neighborhood $V_\epsilon$ such that for all ${\bf g}\in V_\epsilon$, we have:
\begin{equation}
\|{\bf Bg}\|_p^p \geq \|{\bf Bg}\|_1 -  \epsilon p \sum_{i=1}^{LM}|\widetilde{h}_i|^{p-1}|{\bf e}_i^{\mbox{\tiny T}}{\bf Ag }|
\label{eq:bgp}
\end{equation}
On the other hand, the channel identifiability condition \eqref{eq:sufficient} implies that
\begin{equation}
\|{\bf Bg}\|_1\geq \left|\sum_{i=1}^{LM}p|\widetilde{h}_i|^{p-1}\rho_i{\bf e}_i^{\mbox{\tiny T}}{\bf Ag}\right|
\label{eq:bg1}
\end{equation}
Plugging \eqref{eq:bgp} into \eqref{eq:bg1}, we get:
$$
\|{\bf Bg}\|_p^p\geq \left|\sum_{i=1}^{LM}p|\widetilde{h}_i|^{p-1}\rho_i{\bf e}_i^{\mbox{\tiny T}}{\bf Ag}\right|-\epsilon p \sum_{i=1}^{LM}|\widetilde{h}_i|^{p-1}|{\bf e}_i^{\mbox{\tiny T}}{\bf Ag }|.
$$
This proves by using \eqref{eq:ineq_f}, that for all ${\bf g}\in V_\epsilon \cap \widetilde{V}_{\epsilon}$, we have:
$$
\|\widetilde{\bf h}+{\bf Ag}\|_p^p+\|{\bf Bg}\|_p^p\geq \|\widetilde{\bf h}\|_p^p,
$$
thereby proving the result.
\bibliographystyle{IEEEbib}
\bibliography{mybib}

\end{document}